\documentclass[sigconf, nonacm]{acmart}

\usepackage{pvldb}
\usepackage{algorithm}
\usepackage{algpseudocode}

\renewcommand\vldbdoi{XX.XX/XXX.XX}
\renewcommand\vldbpages{XXX-XXX}
\renewcommand\vldbavailabilityurl{https://anonymous.4open.science/r/Einsummable-CE45}

\begin{document}
\title{Every Kernel Is a Join: Automatic Multi-GPU Parallelism for AI
Computations in Einsummable}

%% TODO: author list (double-blind placeholder for now).
\author{Zhimin Ding,
Chen-Kuan Liao,
Chima Adiole,
Brianna Barrow,
Fangzhou Du,
Yu Hsiao,
Ge Huang,
Yicheng Jin,
Ismail Syed, Chris Jermaine}
\affiliation{%
  \institution{Rice University}
  \city{Houston}
  \country{USA}
}
\email{\{zd21,bl86,ca31,bmb15,fd30,rh120,gh31,yj67,is38,cmj4\}@rice.edu}

\begin{abstract}
We present
Einsummable, a prototype system that accepts a PyTorch-like
description of an AI computation and automatically distributes it
across a multi-GPU server. Einsummable models every operation as a relational join followed by
an aggregation over tensor relations, in which the tuples contain
sub-tensors. Each operation exposes its possible decompositions
through \emph{join-agg specs}. An optimizer selects
decompositions across the whole computation to minimize a
communication-cost proxy. Because it searches decompositions rather
than a menu of strategies, Einsummable discovers plans that
mesh-based auto-parallelizers cannot. Each decomposed operation is
implemented by synthesizing an \emph{exchange program}, which is a
topology-aware generalization of Volcano's exchange operator.
Despite
being fully automatic, Einsummable can outperform custom-designed
implementations. For example, on LLaMA transformer blocks on an
eight-GPU A100 server, Einsummable achieves a geometric-mean runtime
of 8.97 ms, versus 13.65 ms for hand-tuned PyTorch and 14.87 ms for
vLLM.
\end{abstract}

\maketitle

\vldbtopmatter

\section{Introduction}

Determining how to shard an AI computation across the GPUs of a multi-GPU server, so that adding GPUs actually reduces the time required to run the computation, is one of the central problems in systems-for-AI \cite{shoeybi2019megatron,jia2019flexflow,zheng2022alpa}. The difficulty is that distributing the operations that make up an AI computation (such as dot-product attention \cite{vaswani2017attention} or matrix multiplication) can incur significant communication overhead. This overhead can be so significant that it is easily possible to distribute a computation across a set of GPUs and find that adding additional GPUs \emph{increases} the running time.

In this paper, we describe a prototype AI system called Einsummable that accepts as input a PyTorch-like \cite{paszke2019pytorch} description of an AI computation and automatically distributes it across the GPUs of a multi-GPU server, with no human expertise required. Einsummable automatically optimizes the \emph{intra-operator parallelism} \cite{zheng2022alpa} of the computation. 
Each individual operation is decomposed so that it runs on all GPUs at once, ideally achieving a speedup close to linear in the number of GPUs. ``Intra-operator parallelism'' is a general class of parallel decompositions that includes data parallelism, tensor (model) parallelism \cite{shoeybi2019megatron,xu2021gspmd}, sequence parallelism \cite{korthikanti2023sequence}, sharded data parallelism \cite{rajbhandari2020zero}, and so on. Crucially, because Einsummable searches a space of decompositions rather than a menu of named strategies, it is not limited to parallelization schemes that have already been invented.\footnote{We do not consider multi-server computations, where the slow server-to-server interconnect means that throughput-enhancing methods such as pipeline parallelism \cite{huang2019gpipe} are more practical. A modern server can present 72 GPUs in a single NVLink domain (e.g., NVIDIA's GB200 NVL72) so single machines cover many use cases.} Optimizing for intra-operator parallelism is fundamentally a database problem. Intra-operator parallelism was the foundation of the shared-nothing parallel database systems of the 1980s and 1990s \cite{dewitt1992parallel}, and its automatic orchestration was the centerpiece of Volcano's exchange operator \cite{graefe1990encapsulation}. Einsummable brings key ideas from database systems to bear on AI workloads.

Einsummable leverages the fact that every operation in a modern AI computation (dot-product attention, softmax, layer normalization, element-wise operations such as the ReLU activation, matrix multiplication, and so on) can be modeled as a relational join followed by a relational aggregation over so-called \emph{tensor relations}, in the spirit of the tensor-relational model of computation \cite{yuan2021tensor}. A ``tensor relation'' is a set of (\texttt{sub\_tensor\_key}, \texttt{sub\_tensor}) pairs used to represent a tensor. For example, a distributed matrix multiply is precisely a join on the shared inner index followed by a \texttt{SUM} aggregation. In the degenerate, purely relational case, the 
\texttt{sub\_tensor\_key} identifies a single cell in the tensor and the \texttt{sub\_tensor} is a scalar value. However, no AI system that manipulates individual scalar values can be competitive for modern, dense AI computations. A competitive system must instead invoke efficient \emph{kernels}, routines (such as a cuBLAS matrix multiply or FlashAttention \cite{dao2022flashattention}) that run on an accelerator (GPU, TPU, etc.) over sub-tensors. 

Thus, the question addressed by Einsummable's design is: How to decompose each AI operation in a larger AI computation into joins and aggregations over tensor relations such that (a) each ``joined'' set of tuples represents enough work that a high-performance kernel can be applied effectively, and (b) the communication overhead introduced by the decomposition is minimized?  Einsummable brings two classical database ideas to the problem of multi-GPU AI computation: the separation of a logical specification from its
physical implementation, and cost-based optimization. 

Existing auto-parallelizing systems such as FlexFlow, GSPMD, and Alpa \cite{jia2019flexflow,xu2021gspmd,zheng2022alpa} search over per-operator sharding annotations drawn from a fixed vocabulary, in which tensor axes are mapped to a logical device mesh. More recent systems such as nnScaler \cite{lin2024nnscaler} let each operator declare which of its dimensions may be partitioned, but every dimension is still partitioned independent of the data. 
Unlike these systems, in Einsummable, each operation defines its own space of relational decompositions. The added expressivity is crucial. The classical 3D matrix-multiplication algorithm, for example, shards all three loops of a matrix multiply simultaneously, and therefore cannot be expressed over a two-dimensional device mesh. Or, consider dot-product attention over packed sequences: the sequence dimension may be split only at the boundaries between sequences, positions that depend on the data rather than on the tensor's shape. And under grouped-query attention, the query-heads dimension cannot be sharded independently, but inherits its sharding from the key/value heads. These constraints are properties of the operator and its input, not of the mesh.

Einsummable uses a three-part solution to the problem of automated decomposition.
First, every operation available to an Einsummable programmer (such as dot-product attention or matrix multiplication) must export the ability to produce what we call \emph{join-agg specs}. Each join-agg spec describes one way to decompose the operation's input tensors into tensor relations so that exactly $p$ tuples are produced by the join that is used to implement the operation, where $p$ is the number of processors; by varying its requests, the system enumerates the operation's decompositions.

Second, because the communication required to move the tuples of a tensor relation from GPU to GPU and to convert one tensor-relational representation into another can be very costly, Einsummable uses a dynamic programming algorithm to automatically choose the most effective set of decompositions for the overall computation, minimizing an estimate of the total number of bytes that must be communicated.

Finally, Einsummable physically realizes the decomposed computation using \emph{exchange programs}. An exchange program is a topology-aware generalization of Volcano's exchange operator \cite{graefe1990encapsulation}, in which all of the aggregation, repartitioning, and replication between producing and consuming operations is expressed in a domain-specific language. Einsummable synthesizes an exchange program for each operation, guided by a simulator of the actual hardware topology, producing a directed acyclic graph of kernel invocations and data-movement operations, which is compiled into a CUDA graph and executed on the GPU server. Notably, Einsummable never invokes a canned communication pattern such as an NCCL all-reduce. Every communication and aggregation pattern is special-purpose, derived during compilation for the specific data movement required.

% Over a suite of experiments measuring LLaMA \cite{grattafiori2024llama} block execution times on a DGX server with eight A100 GPUs, Einsummable, using fully automated decompositions, achieves a geometric mean running time of 8.97 ms, whereas hand-tuned PyTorch requires 13.65 ms and vLLM \cite{kwon2023vllm} (built on top of PyTorch) requires 14.87 ms.

In summary, this paper makes the following contributions:

\begin{itemize}
\item We show that the operations making up modern AI computations can be modeled as joins followed by aggregations over \emph{tensor relations}, giving a single relational framework that captures data, tensor, sequence, and other intra-operator parallelism strategies as special cases (Section~\ref{sec:tensors}).
\item We introduce \emph{join-agg specs} through which each operation exposes its valid decompositions (Section~\ref{sec:tensors}).
\item We present a dynamic-programming optimizer that selects decompositions for an entire computation so as to minimize total communication (Section~\ref{sec:logical}).
\item We introduce \emph{exchange programs} through which each decomposed operation is converted to a directed acyclic graph of kernel invocations and data movements, replacing canned collective libraries with communication synthesized per computation (Section~\ref{sec:physical}).
\item Experimental results show that Einsummable automatically discovers and implements decompositions whose performance meets and sometimes exceeds that of mature systems such as JAX \cite{bradbury2018jax} and PyTorch, even expertly tuned.  (Section~\ref{sec:experiments}).
\end{itemize}
\section{Einsummable Overview}
\label{sec:overview}

At the highest level, Einsummable allows a programmer to create highly-optim\-ized AI computations in a few lines of code.  Einsummable computations run in parallel on multiple AI accelerators without the programmer worrying about how the parallelism happens.  In this section, we give a simple Einsummable example from a programmer's view, and use it to give an overview of the compilation process. Consider a two-layer feed-forward network (FFN),
$\mathbf{Y} = \mathrm{ReLU}(\mathbf{X}\mathbf{W}_1)\mathbf{W}_2,$
where
\[
\mathbf{X} \in \mathbb{R}^{\texttt{tokens} \times \texttt{d\_model}},
\mathbf{W}_1 \in \mathbb{R}^{\texttt{d\_model} \times \texttt{d\_ff}},
\mathbf{W}_2 \in \mathbb{R}^{\texttt{d\_ff} \times \texttt{d\_model}},
\]
and consequently
\[
\mathbf{Y} \in \mathbb{R}^{\texttt{tokens} \times \texttt{d\_model}}.
\]
 
A programmer would construct this computation in Einsummable as
follows:
 
\begin{verbatim}
import einsummable as es
 
graph = es.createGraph()
 
x  = graph.insertTensor(
       es.tensorDesc([tokens, d_model], dtype))
w1 = graph.insertTensor(
       es.tensorDesc([d_model, d_ff], dtype))
w2 = graph.insertTensor(
       es.tensorDesc([d_ff, d_model], dtype))
 
act    = graph.matmul(x, w1)
hidden = graph.unary(act, "relu", 0.0)
y      = graph.matmul(hidden, w2)
 
graph.save(y)
es.compileGraph(graph)
\end{verbatim}
 
The two calls to \texttt{matmul} and the call to \texttt{unary} add
abstract operators to the graph (abstract operators are defined formally in Definition~\ref{def:abstract-operator} of Section~\ref{sec:tensors}). The call to \texttt{graph.save}
identifies an output that must remain available after execution. 
The call to \texttt{compileGraph} asks Einsummable
to map the abstract computation onto hardware. By default, Einsummable
queries the local CUDA runtime and compiles for every visible GPU; the
programmer may instead name a subset of devices, as in
\texttt{es.compileGraph(graph, device=[0, 1, 2, 3])}.
The program contains no device assignment for an
individual operator, no sharding annotations on an individual tensor,
and no explicit communication operations. The number of devices named
in this call is the processor count $p$ over which all join-agg specs
are formed. Join-agg specs are defined in Definition~\ref{def:join-agg-spec} of Section~\ref{sec:tensors}.  Every abstract operator can be asked to produce join-agg specs, which are formal specifications for the tensor-relational decompositions that can be used to implement the operator. During a logical optimization phase, Einsummable considers the join-agg specs available for each abstract operator and
chooses compatible decompositions across the entire FFN so as to minimize communication cost, as described
in Section~\ref{sec:logical}.
 
Once a logical plan has been produced, the compilation process then produces a data-movement schedule among the chosen
devices linking each abstract operator (specified as an \emph{exchange program}, described in Section \ref{sec:exprog}).
The default, \texttt{schedule\_mode="naive"}, uses a naive exchange program. Setting \texttt{schedule\_mode="ai"} asks an
optimizer, guided by a hardware-topology simulator, to search for a
lower-cost schedule of the same decomposed operators
(Section~\ref{sec:transformations}). After compilation, the programmer loads the input tensors (as
\texttt{NumPy} arrays), executes the graph, and offloads the result.

From the programmer's perspective, Einsummable resembles a
conventional tensor graph API. However, unlike other systems,
each operation denotes an abstract operator with many possible
join-agg implementations, and compilation selects devices,
decompositions, and a data-movement schedule without programmer involvement. The remainder
of the paper describes the machinery beneath the API, such as the tensor and
join-agg formalism (Section~\ref{sec:tensors}), the logical optimizer
that chooses decompositions (Section~\ref{sec:logical}), and the
physical optimizer that synthesizes exchange programs
(Section~\ref{sec:physical}).

\section{Tensors and Operators}
\label{sec:tensors}

In this section, we describe the formalisms that underlie our Einsummable implementation.

\subsection{Tensors in Einsummable}

First, we define the notion of \emph{tensors} and \emph{tensor types}.
 
\begin{definition}[Tensor type; tensor]
\label{def:tensor-type}
A \emph{tensor type} $\mathcal{T}$ is a set of functions from index vectors
to real numbers, and is specified by three quantities: a rank
$\mathcal{T}_r$ (a non-negative integer; ``r'' stands for ``rank'', and
matrices are rank-2 tensors), an \emph{upper-bound} vector $\mathcal{T}_u$,
and a \emph{lower-bound} vector $\mathcal{T}_l$, where the bound vectors are
integer vectors of length $\mathcal{T}_r$ satisfying
$\mathcal{T}_l[j] \le \mathcal{T}_u[j]$ for each $j$. Define
$\mathcal{I}(\mathcal{T})$ to be the set
\begin{multline*}
\{\mathcal{T}_l[0], \ldots, \mathcal{T}_u[0]-1\} \times
\{\mathcal{T}_l[1], \ldots, \mathcal{T}_u[1]-1\} \\
\times \cdots \times
\{\mathcal{T}_l[\mathcal{T}_r-1], \ldots, \mathcal{T}_u[\mathcal{T}_r-1]-1\}.
\end{multline*}
This is the set of all indices or keys that obey the upper and lower bounds
of $\mathcal{T}$ (lower bounds are inclusive, upper bounds are exclusive).
$\mathcal{T}$ is then the set of all functions from
$\mathcal{I}(\mathcal{T})$ to the set of real numbers, and a \emph{tensor}
(denoted with an upper-case bold letter such as $\mathbf{U}$) having tensor
type $\mathcal{T}$ is simply one of the functions in $\mathcal{T}$. That is,
we say $\mathbf{U}$ has tensor type $\mathcal{T}$ if
$\mathbf{U} \in \mathcal{T}$.
\end{definition}
 
\emph{Zero-bounded} tensors are those whose tensor type has a lower-bound
vector of all zeros. Einsummable users manipulate zero-bounded tensors, but
internally Einsummable tensor types may have non-zero lower bounds, as it
allows the system to describe/manipulate sub-tensors more easily. 
For human interpretability, zero-bounded tensor types are often specified
using named indices. A named index is a textual string (such as
\texttt{num\_rows} or \texttt{num\_cols}) that corresponds to an upper bound
(for example, we may have $\texttt{num\_rows} = 256$ and
$\texttt{num\_cols} = 512$). Then when we specify a tensor type as
$[\texttt{num\_rows}, \texttt{num\_cols}]$, it corresponds to the tensor
type $\mathcal{T}$ with $\mathcal{T}_r = 2$,
$\mathcal{T}_l = \langle 0, 0 \rangle$, and
$\mathcal{T}_u = \langle 256, 512 \rangle$.

We decompose computations over tensors so that they can be run in parallel.  This is done by decomposing the original computation into smaller computations whose input and output types tile the input and output types to the original computation.

\begin{definition}[Tiling]
Given a tensor type $\mathcal{T}$, we say that the set of tensor types $T = \{\mathcal{T}_1$, $\mathcal{T}_2, \cdots\}$ \emph{tiles} $\mathcal{T}$ if (a) the indices of the tensor types in $T$ cover the indices of $\mathcal{T}$; that is, $\mathcal{I}(\mathcal{T}) =$ $\bigcup_i \mathcal{I} (\mathcal{T}_i)$ and (b) the indices of the tensor types in $T$ are pairwise disjoint; that is, for all $i \neq j$, it is the case that $\mathcal{I}(\mathcal{T}_i) \cap \mathcal{I}(\mathcal{T}_j) = \varnothing$.
\end{definition}

\subsection{Abstract Operators and Join-Agg Specs}

Einsummable programmers build computations using abstract operators over tensors. Abstract operators are operations such as matrix multiplication, flash attention, softmax, and so on.  To a programmer, the actual implementation of the abstract operator (how its inputs are sharded, how it is parallelized) is invisible.  In this subsection we formally define the notion of an abstract operator, as well as how they are implemented.

\begin{definition}[Abstract operator]
\label{def:abstract-operator}
An \emph{abstract operator} in Einsummable is a function over zero-bounded
tensors. Given a list of zero or more zero-bounded input tensor types
$\mathcal{T}^{(i_1)}, \mathcal{T}^{(i_2)}, \ldots$, and a list of zero or
more zero-bounded output tensor types
$\mathcal{T}^{(o_1)}, \mathcal{T}^{(o_2)}, \ldots$, an abstract operator is
a function from the set
$\mathcal{T}^{(i_1)} \times \mathcal{T}^{(i_2)} \times \cdots$ to the set
$\mathcal{T}^{(o_1)} \times \mathcal{T}^{(o_2)} \times \cdots$.
\end{definition}
 
\begin{example}
\label{ex:matmul-signature}
Consider matrix multiplication. If we have named indices \texttt{i},
\texttt{j}, and \texttt{k}, then matrix multiply is a function from
$[\texttt{i}, \texttt{j}] \times [\texttt{j}, \texttt{k}]$ to
$[\texttt{i}, \texttt{k}]$.
\end{example}

Every abstract operator can be decomposed and implemented as a \emph{join}
followed by an \emph{aggregation}, in a way that is described using a
\emph{join-agg spec}. We first need three auxiliary notions.
 
\begin{definition}[Projection; function specification; concrete operator]
\label{def:projection}
We say that a tensor type $\hat{\mathcal{T}}$ is a \emph{projection} of a
tensor type $\mathcal{T}$ if $\hat{\mathcal{T}}_r = \mathcal{T}_r$ and,
componentwise, $\mathcal{T}_l \le \hat{\mathcal{T}}_l$ and
$\hat{\mathcal{T}}_u \le \mathcal{T}_u$; note that then
$\mathcal{I}(\hat{\mathcal{T}}) \subseteq \mathcal{I}(\mathcal{T})$. Given a
tensor $\mathbf{U} \in \mathcal{T}$, the \emph{projection of $\mathbf{U}$
along $\hat{\mathcal{T}}$}, denoted $\hat{\mathbf{U}}$, is the restriction
of $\mathbf{U}$ to the domain $\mathcal{I}(\hat{\mathcal{T}})$. A
\emph{function specification} is an expression of the form
$\mathcal{F} = \hat{\mathcal{T}}^{(i_1)} \times \hat{\mathcal{T}}^{(i_2)}
\times \cdots \rightarrow \hat{\mathcal{T}}^{(o_1)} \times
\hat{\mathcal{T}}^{(o_2)} \times \cdots$, that is, a list of input tensor
types together with a list of output tensor types. Finally, a
\emph{concrete operator} $f$ is a polymorphic function over tensors; we
write $f[\mathcal{F}]$ for the instantiation of $f$ at function
specification $\mathcal{F}$, which is a function from
$\hat{\mathcal{T}}^{(i_1)} \times \hat{\mathcal{T}}^{(i_2)} \times \cdots$
to $\hat{\mathcal{T}}^{(o_1)} \times \hat{\mathcal{T}}^{(o_2)} \times
\cdots$.
\end{definition}

Imagine that we decompose one or more of the tensors that are input into an abstract operator, so that each such input tensor is represented as a set of sub-tensors---that is, as a set of (\texttt{sub\_tensor\_key}, \texttt{sub\_tensor}) pairs, where each key is the tensor type of the corresponding sub-tensor (a \emph{tensor relation} \cite{yuan2021tensor}).  The abstract operator is then implemented in three phases.  If the abstract operator has $n$ inputs, the first phase is a \emph{join} where $n$-tuples of input tensors are formed from the decomposed inputs, and those $n$-tuples are sent to the abstract operator's corresponding concrete operator $f$.  Each invocation of $f$ happens on a different processor (or AI accelerator). The second phase is an \emph{aggregation} where some of the outputs of the various $f$ invocations may need to be combined.  And the final phase is the \emph{recompose}, where the various partial results are stitched together to form the abstract operator's final output.  We now formally define a join-agg spec and its implementation.
 
\begin{definition}[Join-agg spec]
\label{def:join-agg-spec}
Consider an abstract operator of the form
$\mathcal{T}^{(i_1)} \times \mathcal{T}^{(i_2)} \times \cdots \rightarrow
\mathcal{T}^{(o_1)} \times \mathcal{T}^{(o_2)} \times \cdots$. A
\emph{join-agg spec} for this abstract operator is a set
$\mathcal{J} = \{\mathcal{F}_1, \ldots, \mathcal{F}_p\}$ of function
specifications in which each input type $\hat{\mathcal{T}}^{(i_j)}$ is a
projection of $\mathcal{T}^{(i_j)}$ and each output type
$\hat{\mathcal{T}}^{(o_j)}$ is a projection of $\mathcal{T}^{(o_j)}$, and
which satisfies the following \emph{tiling condition}: for each output
index $j$, the distinct $j$th output types appearing in
$\mathcal{J}$ tile
$\mathcal{T}^{(o_j)}$. When a join-agg spec is executed, each
of its $p$ function specifications is assigned to its own processor.
\end{definition}
 
\begin{definition}[Implementation]
\label{def:implements}
We say that a join-agg spec $\mathcal{J}$ \emph{implements} an abstract
operator if there exist a concrete operator $f$ and a binary, associative,
and commutative \emph{aggregation operation} $\oplus$ over the reals
(applied elementwise to tensors of a common type) such that, for every
input to the abstract operator, the following three-step procedure recovers the
correct output:
\begin{enumerate}
\item \textbf{Join.} For each $\mathcal{F} \in \mathcal{J}$: project each
input tensor $\mathbf{U}^{(i_j)}$ along the corresponding
$\hat{\mathcal{T}}^{(i_j)}$ from $\mathcal{F}$ to obtain
$\hat{\mathbf{U}}^{(i_j)}$, then execute $f[\mathcal{F}]$ on
$\hat{\mathbf{U}}^{(i_1)}, \hat{\mathbf{U}}^{(i_2)}, \ldots$ to obtain
$\hat{\mathbf{U}}^{(o_1)}, \hat{\mathbf{U}}^{(o_2)}, \ldots$, where each
$\hat{\mathbf{U}}^{(o_j)} \in \hat{\mathcal{T}}^{(o_j)}$.
\item \textbf{Aggregate.} For each output index $j$ and each distinct type
$\hat{\mathcal{T}}$ appearing as a $j$th output type in $\mathcal{J}$,
combine using $\oplus$ the $j$th outputs of all function specifications
whose $j$th output type is $\hat{\mathcal{T}}$, yielding a single tensor
$\hat{\mathbf{V}} \in \hat{\mathcal{T}}$.
\item \textbf{Recompose.} For each output index $j$, the tiling condition
ensures that there is a unique tensor
$\mathbf{V}^{(o_j)} \in \mathcal{T}^{(o_j)}$ that agrees with each
aggregated tensor from step (2) on that tensor's index set; this
$\mathbf{V}^{(o_j)}$ is the $j$th output.
\end{enumerate}
\end{definition}
 
\begin{example}[Matrix multiply, decomposed]
\label{ex:matmul-join-agg}
Let $\texttt{i} = \texttt{k} = 256$ and $\texttt{j} = 512$, and consider
the matrix-multiply abstract operator of
Example~\ref{ex:matmul-signature}, from
$[\texttt{i}, \texttt{j}] \times [\texttt{j}, \texttt{k}]$ to
$[\texttt{i}, \texttt{k}]$. (This may be read as the first matrix
multiply of the FFN of Section~\ref{sec:overview}, with
$\texttt{i} = \texttt{tokens}$, $\texttt{j} = \texttt{d\_model}$, and
$\texttt{k} = \texttt{d\_ff}$.) A simple join-agg spec for this
operator splits the shared index \texttt{j} into two halves. Let
$\mathcal{J} = \{\mathcal{F}_1, \mathcal{F}_2\}$, where
$\mathcal{F}_1 = \mathcal{A} \times \mathcal{B} \rightarrow \mathcal{C}$
and
$\mathcal{F}_2 = \mathcal{A}' \times \mathcal{B}' \rightarrow \mathcal{C}$,
where $\mathcal{A}$ and $\mathcal{A}'$ are the left and right halves
(columns $[0, 256)$ and $[256, 512)$) of the first input type,
$\mathcal{B}$ and $\mathcal{B}'$ are the top and bottom halves of the
second, and $\mathcal{C}$ is the entire output type
$[\texttt{i}, \texttt{k}]$. Each of these is a projection of the
corresponding input or output type of the abstract operator, and since
$\mathcal{C}$ is the only output type appearing in $\mathcal{J}$, with
$\mathcal{I}(\mathcal{C}) = \mathcal{I}([\texttt{i}, \texttt{k}])$, the
tiling condition holds trivially.
 
This join-agg spec implements matrix multiply. Take $f$ to be the
polymorphic block matrix multiply, which for a function specification
$\hat{\mathcal{A}} \times \hat{\mathcal{B}} \rightarrow \hat{\mathcal{C}}$
multiplies the two projected blocks, contracting over the index range
of $\hat{\mathcal{A}}$'s second dimension, and take $\oplus$ to be addition. In step (1) of
Definition~\ref{def:implements}, the processor assigned $\mathcal{F}_1$
computes the partial products over the first half of \texttt{j}, and the
processor assigned $\mathcal{F}_2$ over the second half. In step (2), the
two results have the identical output type $\mathcal{C}$ and are summed
elementwise. Step (3) is not necessary, as $\mathcal{C}$ is the only output
block.
 
We can further decompose so as to require recomposition. Additionally split
\texttt{i} into two halves. The resulting join-agg spec has four function
specifications $\mathcal{F}_{m,n}$ for $m, n \in \{1, 2\}$, where
$\mathcal{F}_{m,n}$ multiplies the block of the first input with rows in
the $m$th half of \texttt{i} and columns in the $n$th half of \texttt{j}
by the block of the second input with rows in the $n$th half of
\texttt{j}, and has the single output type $\mathcal{C}_m$, with
$(\mathcal{C}_1)_l = \langle 0, 0 \rangle$,
$(\mathcal{C}_1)_u = \langle 128, 256 \rangle$,
$(\mathcal{C}_2)_l = \langle 128, 0 \rangle$, and
$(\mathcal{C}_2)_u = \langle 256, 256 \rangle$. Now two distinct output
types appear in the spec. In
step (2), the outputs of $\mathcal{F}_{m,1}$ and $\mathcal{F}_{m,2}$ are
summed to produce the aggregated block of type $\mathcal{C}_m$, and in
step (3) the two aggregated blocks, which tile the output index set, are
stitched into the final output tensor.
\end{example}

\subsection{Shardable Indices}

One or more of the named indices in the specification for an abstract operator are \emph{shardable}.  Shardable indices are used by Einsummable to request join-agg specs from an abstract operator.  Every shardable index will either have an associated list of positions that it can be sharded on, or else it will have an annotation \texttt{any} indicating that it can be sharded at any position.  To request a join-agg spec from an abstract operator, Einsummable decomposes each shardable index into one or more intervals and provides the intervals to the abstract operator, which returns a corresponding join-agg spec.  If $\texttt{intervals}[s]$ is the number of intervals associated with the $s$th shardable index, then the returned join-agg spec must contain exactly $\prod_s \texttt{intervals}[s]$ function specifications, one for each combination of intervals.
Crucially, since the abstract operator is itself responsible for building the join-agg spec, there is a great deal of flexibility in how operators can be decomposed.  

\begin{example}[Matrix multiply]
\label{ex:matmul-sharding}
For the matrix multiply of Example~\ref{ex:matmul-signature}, all three
named indices \texttt{i}, \texttt{j}, and \texttt{k} are shardable with
the annotation \texttt{any}, since a matrix multiply may be decomposed
between any two rows or columns of its inputs. Now let
$\texttt{i} = \texttt{k} = 256$ and $\texttt{j} = 512$. Suppose that
during optimization, Einsummable decomposes \texttt{i} into the intervals
$[0, 128)$ and $[128, 256)$, \texttt{j} into the intervals
$[0, 256)$ and $[256, 512)$, and leaves \texttt{k} as the single interval
$[0, 256)$. Since $\texttt{intervals}[\texttt{i}] = 2$,
$\texttt{intervals}[\texttt{j}] = 2$, and
$\texttt{intervals}[\texttt{k}] = 1$, the abstract operator must return a
join-agg spec containing $2 \cdot 2 \cdot 1 = 4$ function specifications, and it returns the join-agg spec
$\{\mathcal{F}_{m,n}\}$ of Example~\ref{ex:matmul-join-agg}, where
$\mathcal{F}_{m,n}$ corresponds to the $m$th interval of \texttt{i} paired
with the $n$th interval of \texttt{j}. Note the
differing roles of the indices. \texttt{j} indexes the contracted
dimension, so specifications that share an \texttt{i} interval but differ
in their \texttt{j} interval are combined by
$\oplus = +$ during aggregation, whereas the decomposition of \texttt{i}
passes through to an output tiling.
\end{example}

\begin{example}[Flash attention]
\label{ex:flash-sharding}
Consider flash-attention with grouped-query
attention, with input types
$[\texttt{heads\_q}$, \texttt{sequence}, $\texttt{d\_head}]$ (the queries),
$[\texttt{heads\_kv}$, \texttt{sequence}, $\texttt{d\_head}]$ (the keys),
and $[\texttt{heads\_kv},$ \texttt{sequence}, $\texttt{d\_head}]$ (the
values), and with output type
$[\texttt{heads\_q}, \texttt{sequence}, \texttt{d\_head}]$. Let
$\texttt{heads\_q} = 32$ and $\texttt{heads\_kv} = 8$, so that query
heads $4m, \ldots, 4m + 3$ are associated with the key-value head $m$. Of these
named indices, only \texttt{heads\_kv} and \texttt{sequence} are
shardable.
\texttt{heads\_q} is \emph{not} directly shardable, yet sharding
\texttt{heads\_kv} has the effect of sharding \texttt{heads\_q}. If
Einsummable decomposes \texttt{heads\_kv} into its eight unit intervals,
the abstract operator returns a join-agg spec with eight function
specifications, where the $m$th specification projects the key and value
inputs onto key-value head $m$ and the query input and
output onto the query-head interval $[4m, 4(m+1))$. This
logic is handled by the abstract operator when it produces the
join-agg spec. This is why join-agg spec production is left to individual abstract operators. Decomposition logic can be operator specific and the
Einsummable system does not need to directly reason about them.

The index \texttt{sequence} is an example of a shardable index without the \texttt{any} annotation. Suppose $\texttt{sequence} = 4096$
holds three independent sequences packed end to end, ending at positions
$1024$, $2560$, and $4096$. Attention never crosses a sequence boundary,
but positions within a single sequence will attend to one another, so
\texttt{sequence} may be sharded only at the ends of individual
sequences. 
% Finally, note that in every
% join-agg spec above, the output types are fully distinct blocks that tile the
% output index set. There is no overlap, so no
% aggregation phase is required. This is in contrast to matrix multiplication. Flash
% attention is decomposed into a pure join and recompose.
\end{example}

\section{Logical Optimization}
\label{sec:logical}

\subsection{Overview}

Recall from Section~\ref{sec:overview} that compilation proceeds in two phases, which roughly correspond to logical and physical optimization in classical database systems. This section describes logical optimization, where Einsummable associates a join-agg spec with each abstract operator, chosen so as to minimize the communication associated with the computation. The result is a logical plan.

\subsection{Choosing Join-Agg Specs}

A computation is a directed, acyclic graph where vertices are abstract operators, and there is an edge $(\texttt{A}, \texttt{B})$ if abstract operator \texttt{B} consumes the output of \texttt{A}.  For the time being we assume that each output is consumed only once, and that each operator produces only a single output.
Logical optimization is powered by a dynamic programming algorithm
operating over tilings. Let $\texttt{Opt}(\texttt{Op}, T)$ denote the lowest
communication cost possible for computing the output of abstract operator
\texttt{Op}, subject to the constraint that the tiling of the output is
$T$.
Assume we have a cost function
$\texttt{Cost}(\mathcal{J},$ $T_1, T_2, \cdots)$ that computes the
communication cost required to implement join-agg spec $\mathcal{J}$ given
that the first input has tiling $T_1$, the second $T_2$, and so on. This
includes the cost of repartitioning each input from its tiling into the
input projections required by $\mathcal{J}$, as well as the cost of the
join and the aggregation; we describe the details
 subsequently. Further, for a join-agg spec $\mathcal{J}$, let
$\texttt{out}(\mathcal{J})$ denote the tiling of the output induced by
$\mathcal{J}$ (that is, the set of distinct output types appearing in
$\mathcal{J}$ that tiles the output).

A key observation is that in $\texttt{Opt}(\texttt{Op}, T)$, the tiling
$T$ need only range over the output tilings that \texttt{Op}'s join-agg
specs actually produce. It is never better to first
repartition an output into some intermediate tiling and then repartition
again at the consumer than it is to repartition directly. Consequently,
nothing is lost by restricting the entries of the dynamic programming
table for \texttt{Op} to the tilings
$\{ \texttt{out}(\mathcal{J}) : \mathcal{J} \in
\texttt{specs}(\texttt{Op}) \}$. Note that this restriction does not determine
the join-agg spec as several specs may produce the same output tiling. For example, in
the matrix multiply of Example~\ref{ex:matmul-sharding},
specs that differ only in how the contracted index \texttt{j} is
decomposed all induce the same tiling of the output.

Consider an abstract operator \texttt{Op} whose inputs are produced by
abstract operators $\texttt{Op}_1, \ldots, \texttt{Op}_n$, and let
$\texttt{specs}(\texttt{Op})$ denote the set of join-agg specs over $p$
processors that \texttt{Op} can produce. This set is obtained by
requesting one join-agg spec from \texttt{Op} for each decomposition of
its shardable indices into intervals whose counts multiply to $p$. The
recurrence is
\begin{multline*}
\texttt{Opt}(\texttt{Op}, T) \;=\;
\min_{\substack{\mathcal{J} \in \texttt{specs}(\texttt{Op})
\,:\, \texttt{out}(\mathcal{J}) = T \\
T_1, \ldots, T_n}}
\Big[\,
\sum_{k=1}^{n} \texttt{Opt}(\texttt{Op}_k, T_k) \\
\;+\; \texttt{Cost}(\mathcal{J}, T_1, \ldots, T_n)
\,\Big],
\end{multline*}
where each $T_k$ ranges over the candidate tilings of the corresponding
producer $\texttt{Op}_k$. That is, to produce the output of \texttt{Op}
tiled according to $T$, choose a join-agg spec $\mathcal{J}$ for
\texttt{Op} that induces the tiling $T$, and choose a tiling $T_k$ for
each input; pay each producer $\texttt{Op}_k$ to deliver its output tiled
according to $T_k$, and pay to execute $\mathcal{J}$ on inputs that are so tiled
(including the cost of repartitioning each input from $T_k$ into the
blocks that $\mathcal{J}$ requires). The base case is a
tensor that is an input to the computation. Such a tensor may be
materialized across the processors in any tiling, so
$\texttt{Opt}(\texttt{Op}, T) = 0$ for every $T$ when \texttt{Op} is an
input. 

The cost of
the best logical plan is then
$\min_T \texttt{Opt}(\texttt{Op}^{*}, T)$, where $\texttt{Op}^{*}$ is the
operator producing the computation's designated output tensor, and the
plan itself (a choice of join-agg spec for every abstract operator)
is recovered by tracing back the minimizing choices.

%Two observations make the algorithm practical. First, since every output
%is consumed exactly once, the computation is a tree rooted at
%$\texttt{Op}^{*}$; the subproblems $\texttt{Opt}(\texttt{Op}_k, \cdot)$
%are therefore independent of one another, the recurrence exhibits optimal
%substructure, and $\texttt{Opt}$ can be evaluated bottom-up (or memoized),
%visiting each (operator, tiling) pair only once. Second, 

To fill the
table for an operator, it is enough to compute, once per spec,
\[
B(\mathcal{J}) = \min_{T_1, \ldots, T_n}
\Big[ \sum_{k=1}^{n} \texttt{Opt}(\texttt{Op}_k, T_k)
+ \texttt{Cost}(\mathcal{J}, T_1, \ldots, T_n) \Big],
\]
after which $\texttt{Opt}(\texttt{Op}, T)$ is the minimum of
$B(\mathcal{J})$ over the specs with $\texttt{out}(\mathcal{J}) = T$.
Since the table for \texttt{Op} has one entry per distinct 
tiling, the number of entries is at most the number of specs. If each abstract operator has at
most $s$ join-agg specs and at most $n$ inputs, so that each tensor has
at most $t \le s$ candidate tilings, the total running time is
$O(\sum_{\texttt{Op}} s \cdot t^{n})$, where the $t^{n}$ factor comes
from enumerating combinations of input tilings.

Algorithm~\ref{alg:logical-opt} gives the complete procedure. It
processes the operators bottom-up, computing $B(\mathcal{J})$ for each
spec. For each output tiling, it records the cheapest join-agg spec that
induces it, together with the input tilings that achieved the minimum.
A final backwards pass walks the tree from the designated output.

\begin{algorithm}[h]
\caption{Logical optimization}
\label{alg:logical-opt}
\begin{algorithmic}[1]
\Require a computation tree whose leaves are input tensors and whose output is produced by $\texttt{Op}^{*}$, a processor
count $p$
\Ensure a join-agg spec $\mathit{spec}[\texttt{Op}]$ for every abstract
operator \texttt{Op}
\For{each input tensor $\mathbf{X}$ of the computation}
  \State $\mathit{cand}[\mathbf{X}] \gets \{\bot\}$;\quad
         $\texttt{Opt}[\mathbf{X}, \bot] \gets 0$
\EndFor
\For{each operator \texttt{Op}, in bottom-up (topological)
     order}
  \State let $\texttt{Op}_1, \ldots, \texttt{Op}_n$ be the producers of
         the $n$ inputs of \texttt{Op}
  \State $\mathit{cand}[\texttt{Op}] \gets \varnothing$
  \For{each $\mathcal{J} \in \texttt{specs}(\texttt{Op})$}
    \State choose $(T_1^{\star}, \ldots, T_n^{\star}) \in
           \mathit{cand}[\texttt{Op}_1] \times \cdots \times
           \mathit{cand}[\texttt{Op}_n]$ to
    \Statex \hspace{\algorithmicindent}\hspace{\algorithmicindent}
           min $B \gets \sum_{k=1}^{n} \texttt{Opt}[\texttt{Op}_k, T_k]
           + \texttt{Cost}(\mathcal{J}, T_1, \ldots, T_n)$
    \State $T \gets \texttt{out}(\mathcal{J})$
    \If{$T \notin \mathit{cand}[\texttt{Op}]$ \textbf{ or }
        $B < \texttt{Opt}[\texttt{Op}, T]$}
      \State $\mathit{cand}[\texttt{Op}] \gets
             \mathit{cand}[\texttt{Op}] \cup \{T\}$
      \State $\texttt{Opt}[\texttt{Op}, T] \gets B$
      \State $\mathit{choice}[\texttt{Op}, T] \gets
             (\mathcal{J}, T_1^{\star}, \ldots, T_n^{\star})$
    \EndIf
  \EndFor
\EndFor
\State $T^{\star} \gets$ the $T \in \mathit{cand}[\texttt{Op}^{*}]$
       minimizing $\texttt{Opt}[\texttt{Op}^{*}, T]$
\State \Call{Assign}{$\texttt{Op}^{*}, T^{\star}$}
\Statex
\Procedure{Assign}{$\texttt{Op}, T$}
  \State $(\mathcal{J}, T_1, \ldots, T_n) \gets
         \mathit{choice}[\texttt{Op}, T]$
  \State $\mathit{spec}[\texttt{Op}] \gets \mathcal{J}$
  \For{$k = 1, \ldots, n$}
    \If{the $k$th input of \texttt{Op} is produced $\texttt{Op}_k$}
      \State \Call{Assign}{$\texttt{Op}_k, T_k$}
    \EndIf
  \EndFor
\EndProcedure
\end{algorithmic}
\end{algorithm}

\subsection{Cost Model}
\label{sec:cost-model}

Our implementation of $\texttt{Cost}(\mathcal{J}, T_1, \ldots, T_n)$ does not take into account computation cost, as for every abstract operator available in the standard toolkit, the computational cost of implementing each join-agg spec for the operator is the same in terms of the number of floating point operations required. As Einsummable requests decompositions into near-equal intervals, the candidate join-agg specs are load-balanced and so the computation time is effectively the same across the candidate specs.

However, different join-agg specs can have very different communication costs. Thus \texttt{Cost} computes a proxy for communication, just as a logical query optimizer in a relational database system costs candidate plans using a proxy such as the number of intermediate tuples produced, deferring precise, hardware-aware costing to physical optimization.  Our \texttt{Cost} function estimates the number of floating point numbers that may need to be transferred to run the join-agg spec $\mathcal{J}$, ignoring which processor each block of data actually resides on; data placement is considered during physical optimization.  Assuming that each abstract operator produces a single output, the estimate includes three separate sub-costs:

\begin{enumerate}
    \item \textbf{Repartition.}  Consider the $k$th input tensor to the operator in question, with type $\mathcal{T}^{(i_k)}$, and let $\texttt{in}_k(\mathcal{J})$ denote the set of distinct $k$th input types appearing in the function specifications of $\mathcal{J}$.  If $\texttt{in}_k(\mathcal{J}) \neq T_k$, then add $|\mathcal{I}(\mathcal{T}^{(i_k)})|$ to the overall cost.  Intuitively: if we have to repartition the tensor, we may have to move every floating point number in the tensor around the network; there are $|\mathcal{I}(\mathcal{T}^{(i_k)})|$ such numbers. Note that if the $k$th input is an input to the computation the tensor is materialized directly into the blocks required by $\mathcal{J}$, and this sub-cost is zero.

    \item \textbf{Join.}  Consider each function specification in $\mathcal{J}$ of the form $\hat{\mathcal{T}}^{(i_1)} \times \hat{\mathcal{T}}^{(i_2)} \times \cdots \rightarrow
\hat{\mathcal{T}}^{(o)}$.  Simply add $\sum_k |\mathcal{I}(\hat{\mathcal{T}}^{(i_k)})|$ to the overall cost.  Intuitively: we have to form every $n$-tuple that is the input to some invocation of the concrete operator $f$, and this counts the total number of floating point numbers in those $n$-tuples.  This penalizes join-agg specs that require large inputs to $f$.

\item \textbf{Aggregation.} Consider each distinct type
$\hat{\mathcal{T}}^{(o)}$ appearing as an output type in $\mathcal{J}$.  Count the number of times this distinct type appears in $\mathcal{J}$; let this be $m$.  Add $(m - 1)|\mathcal{I}(\hat{\mathcal{T}}^{(o)})|$ to the overall cost.  Intuitively: these $m$ copies need to be aggregated down to one copy, which (naively) requires moving $(m-1)$ copies to the location of the last copy, and this counts the total number of floating point numbers moved. 

\end{enumerate}

\subsection{Detailed Example: 3D Matrix Multiply}

\label{sec:3d-matmul}
We now use this machinery end to end, and show that the optimizer
recovers the classical ``3D'' matrix-multiplication
algorithm~\cite{dekel1981parallel,agarwal1995threed}, which is
communication-optimal for square
matrices~\cite{irony2004communication}.

Consider a computation consisting of a single matrix multiply
(Example~\ref{ex:matmul-signature}) with
$\texttt{i} = \texttt{j} = \texttt{k} = N$, whose two input matrices are
inputs to the computation, to be run on $p$ processors; for simplicity,
assume $p = q^{3}$ for an integer $q$, and that the interval counts under
consideration divide $N$. As in
Example~\ref{ex:matmul-sharding}, every candidate join-agg spec for this
operator decomposes \texttt{i}, \texttt{j}, and \texttt{k} into $a$, $b$,
and $c$ equal intervals with $abc = p$, and contains one function
specification per cell of the resulting $a \times b \times c$ grid. The
specification for a given cell multiplies an $(N/a) \times (N/b)$ block
of the first input by an $(N/b) \times (N/c)$ block of the second
producing a partial result for an $(N/a) \times (N/c)$ block of the
output.

We evaluate $\texttt{Cost}$ on such a spec. Sub-cost (1) is zero as both
inputs are inputs to the computation. For sub-cost (2) each of the $abc$ function
specifications has input blocks totaling
$\frac{N^{2}}{ab} + \frac{N^{2}}{bc}$ numbers, for a total of
$abc \left( \frac{N^{2}}{ab} + \frac{N^{2}}{bc} \right) = (a + c) N^{2}$. Intuitively, each element of the first input is replicated $c$ times, and
each element of the second input $a$ times. For sub-cost (3), the
distinct output types are the $ac$ blocks of an $a \times c$ tiling of
the output, and each appears $m = b$ times, contributing
$(b - 1) N^{2}$. Hence
\[
\texttt{Cost} \;=\; (a + c)\, N^{2} + (b - 1)\, N^{2}
\;=\; (a + b + c - 1)\, N^{2}.
\]
By the arithmetic/geometric mean inequality,
$a + b + c \ge 3 \sqrt[3]{abc} = 3q$, with equality exactly when
$a = b = c = q$. The optimizer therefore selects the
$q \times q \times q$ grid, at a cost of
$(3q - 1) N^{2} = (3 p^{1/3} - 1) N^{2}$ (each processor multiplies an
$(N/q) \times (N/q)$ block of the first input by an
$(N/q) \times (N/q)$ block of the second, and the $q$ partial results
for each output block are summed). This is precisely the 3D
matrix-multiplication algorithm~\cite{agarwal1995threed}.

%Finally, note that in this example, the proxy nature of \texttt{Cost} is
%harmless. Sub-cost (2) charges every input block, including blocks that,
%under a fortunate placement, would already reside on the processor that
%consumes them. Under the best possible placement, the true transfer
%volume is $(c - 1) N^{2} + (a - 1) N^{2} + (b - 1) N^{2}$ --- exactly
%$2 N^{2}$ less than $\texttt{Cost}$, a difference that does not depend on
%$(a, b, c)$. The proxy therefore ranks all candidate specs in the same
%order as the true transfer volume, and the $q \times q \times q$ grid is
%the true minimizer as well.

\subsection{Edge Cases}
\label{sec:edge-cases}

The discussion so far assumes that each abstract operator has one
output, consumed by one operator. Three extensions cover the cases
that arise in practice. If an operator has several outputs that all
flow to the same consumer, \texttt{Opt} is indexed by a \emph{tuple}
of tilings, one per output, and the consumer selects a single table
entry of that producer rather than choosing the tilings independently;
the graph remains a tree and optimal substructure is preserved. If an
output has several consumers, the computation is no longer a tree, and
an exact algorithm would have to track the joint demands on every
shared operator, a state space that grows exponentially with the
amount of sharing. Instead, the graph is cut at every edge out of a
multiply-consumed tensor, and the resulting trees are optimized with
Algorithm~\ref{alg:logical-opt} one at a time in decreasing order of size, each subject to the join-agg specs
already fixed (a tensor produced in an already-optimized tree enters
the recurrence like a computation input, except that \texttt{Cost}
charges any repartitioning). Finally, because the per-operator work is
$O(s \cdot t^{n})$ in the number of inputs $n$, our implementation
applies the same tactic to linear paths of operators: each operator on
a path draws at most one input from within the path, so the inner
minimization in $B(\mathcal{J})$ ranges over a single producer's
tilings and the work drops to $O(s \cdot t)$. In both cases the plan
is optimal within each piece, though not necessarily globally.

%% physical_optimization.tex
%% Supersedes exchange_programs.tex.
%% Requires the same preamble as the rest of the paper:
%%   \usepackage{amsmath, amssymb, amsthm}
%%   \theoremstyle{definition}
%%   \newtheorem{definition}{Definition}
%%   \newtheorem{example}[definition]{Example}
%%   \newtheorem{theorem}[definition]{Theorem}
%% Programs are set in verbatim; no additional packages are needed.

\section{Physical Optimization}
\label{sec:physical}

\subsection{Overview}
\label{sec:physical-overview}

Einsummable's physical optimizer accepts the output of
the logical optimizer (a DAG of abstract operators,
each labeled with a selected join-agg spec) and produces an
executable artifact. In the current version of Einsummable, this
artifact is an NVIDIA CUDA graph, in which each vertex is a GPU kernel
or a network operation, which can be executed directly on an NVIDIA
server.

Physical optimization proceeds through the graph of abstract operators
in topologically sorted order, so that when an operator \texttt{op} is
reached, every operator supplying one of its inputs has already been
compiled. Each producer is compiled \emph{through its join} so that when \texttt{op} is reached, the unaggregated
partial blocks of each of \texttt{op}'s inputs are pinned to specific
GPUs. It is \texttt{op}'s job to aggregate those partial blocks as it
brings them where they are needed.

To do so, the compiler generates the \emph{preamble} of a program in a
domain-specific language called \texttt{ExProg} (for \emph{exchange
program}). This includes an \texttt{IN} declaration recording the location and size
of every partial block, and an \texttt{OUT} declaration, derived from
\texttt{op}'s join-agg spec, requiring each aggregated input block at
the sites that will run \texttt{op}'s concrete operator. The sites in
the \texttt{OUT} declaration are expressed as partitions declared by
size, so that the placement of \texttt{op}'s $p$ concrete-operator
instances is itself chosen during synthesis. Exchange programs are a
topology-aware generalization of Volcano's exchange
operator~\cite{graefe1990encapsulation}. As with the exchange operator, an
exchange program is where repartitioning, replication, and inter-processor
parallelism are visible. Unlike the exchange operator, an exchange program also
performs aggregation, and it is written against an
explicit model of the hardware topology.

Given the preamble, the body of the exchange program is synthesized so
as to minimize the execution time predicted by a hardware simulator.
The compiler uses the synthesized program and its simulated
execution trace (fixing the dependency structure among transfers,
aggregations, and compute kernels) to create the portion of the
CUDA graph that links \texttt{op} to its inputs and invokes each
of the $p$ instances of the concrete operator $f$ implementing
\texttt{op}.

\subsection{The \texttt{ExProg} Language}
\label{sec:exprog}

\paragraph{State.} An exchange program manipulates named
\emph{atoms} which are tensors, each with a size (in GB). The state of the
machine at any instant is a set of pairs $(\textit{aggregate},
\textit{site})$, where an aggregate is a $\oplus$-combination of atoms,
written \texttt{A+B+C}, and a site is a GPU. Because $\oplus$ is
associative and commutative, \texttt{A+B} and \texttt{B+A} denote the
same aggregate; the symbol \texttt{+} denotes the aggregation operation
$\oplus$ of the operator at hand, which need not be addition. A program
declares an initial state and a required final state, and its body must
transform the former into the latter.

\begin{itemize}
\item \texttt{GPUS:} lists the valid sites, and \texttt{PARTITIONS:}
optionally names groups of sites; the groups must be disjoint and cover
the declared GPUs. A partition may
also be declared by size alone, as in \texttt{G1 \{size 4\}}, in which
case the compiler chooses its membership, and the final state must hold for some choice of membership.
\item \texttt{IN:} places each atom at a site, optionally with
a size, e.g., \texttt{IN: A @ 0 (size: 1.0 GB)} (sizes default
to 1 GB).
\item \texttt{OUT:} lists the aggregates that must exist, and at which
sites or partitions, e.g., \texttt{OUT: A+B @ 0, 1}.
\end{itemize}

\paragraph{Actions.} The body of a program is a set of \emph{actions},
each naming an aggregate, the site at which it is formed, and an
optional destination list. \texttt{A+B @ 0 -> 1} aggregates \texttt{A}
and \texttt{B} at site 0 and sends the result to site 1.
\texttt{A @ 0 -> 1, 2, 3} multicasts \texttt{A} to three sites.
\texttt{A+B @ 0} aggregates locally with no network traffic. \texttt{A+B @ G1} aggregates at every site of partition \texttt{G1}
simultaneously. When convenient, we write several actions on one
line, separated by semicolons. There are three key rules. First, \emph{a program is a set, not a sequence}. An action
becomes eligible the moment every atom of its aggregate is present at
its source, and all eligible actions proceed concurrently. Intermediate states emerge from dataflow, not program
order. Second, every action with an arrow is a two-stage pipeline: a
local aggregation (reads and writes at HBM bandwidth), followed by a
transfer across the link. Third, aggregating at the source
materializes the aggregate there, satisfying the final-state
requirement at that site, and other actions may reuse it.

\paragraph{Cost.} A hardware simulator assigns each program a completion
time. DGX-A100 and DGX-H100 machines are modeled as fully non-blocking
all-to-all NVSwitch fabrics in which the only contention is at
endpoints. A multicast to $k$ destinations divides the source's egress
bandwidth by $k$. The DGX-V100 is far more restrictive as its eight GPUs
form a hybrid cube-mesh with 50 GB/s double NVLinks on some pairs,
25 GB/s single links on others, and no direct link on the rest
so transfers between unlinked pairs fall back to a PCIe bridge with
15.75 GB/s of total capacity shared by all such transfers.

\begin{example}[A peer-to-peer swap]
\label{ex:exchange-swap}
Two GPUs each hold one atom, and both must end up with the aggregate:
\begin{verbatim}
IN:  A @ 0; B @ 1
OUT: A+B @ 0, 1

A @ 0 -> 1;  B @ 1 -> 0;  A+B @ 0, 1
\end{verbatim}
Both transfers are eligible immediately and proceed concurrently and when complete, the aggregations run. On a DGX-H100 with 1 GB
atoms, the simulator reports completion in roughly 3.1 milliseconds with
2.2 ms for the (overlapping) transfers, and the remainder for the local
aggregations.
\end{example}

\begin{example}[Recursive doubling]
\label{ex:exchange-doubling}
Four GPUs each hold one atom, and all four must obtain the total
aggregate:
\begin{verbatim}
IN:  A @ 0; B @ 1; C @ 2; D @ 3
OUT: A+B+C+D @ 0, 1, 2, 3

A @ 0 -> 1;    B @ 1 -> 0
C @ 2 -> 3;    D @ 3 -> 2
A+B @ 0 -> 2;  A+B @ 1 -> 3
C+D @ 2 -> 0;  C+D @ 3 -> 1
A+B+C+D @ 0, 1, 2, 3
\end{verbatim}
The first four actions are eligible at time zero. The second group
becomes eligible as its operands arrive. Partial aggregates cross the network in
the second round, halving the traffic relative to gathering all four
atoms everywhere. On a DGX-H100 with 1 GB atoms, the simulator reports
completion in 6.2 milliseconds.
\end{example}

\begin{example}[Topology-aware placement on a DGX-V100]
\label{ex:exchange-v100}
The following specification leaves placement to the compiler: the two
partitions are declared by size alone, so the final state must hold
only for some choice of two four-GPU groups.
\begin{verbatim}
GPUS: 0, 1, 2, 3, 4, 5, 6, 7
PARTITIONS: G1 {size 4} G2 {size 4}
IN:  A @ 0; B @ 1; C @ 2; D @ 3;
     E @ 4; F @ 5; G @ 6; H @ 7
OUT: A+B+E+F @ G1; C+D+G+H @ G2
\end{verbatim}
On a DGX-V100, the compiler chooses $\texttt{G1} = \{0, 1, 4, 5\}$ and
$\texttt{G2} = \{2, 3, 6, 7\}$, aligning each group with the machine's
50 GB/s double-links ($0{\leftrightarrow}4$,
$1{\leftrightarrow}5$, $2{\leftrightarrow}6$, $3{\leftrightarrow}7$),
and emits:
\begin{verbatim}
A @ 0 -> 1;    B @ 1 -> 0;    C @ 2 -> 3
D @ 3 -> 2;    E @ 4 -> 5;    F @ 5 -> 4
G @ 6 -> 7;    H @ 7 -> 6
A+B @ 0 -> 4;  A+B @ 1 -> 5;  C+D @ 2 -> 6
C+D @ 3 -> 7;  E+F @ 4 -> 0;  E+F @ 5 -> 1
G+H @ 6 -> 2;  G+H @ 7 -> 3
A+B+E+F @ G1;  C+D+G+H @ G2
\end{verbatim}
The first round of swaps
uses the 25 GB/s single links within each CPU socket. The second round
crosses the socket boundary on the 50 GB/s pillars and no
transfer uses the 15.75 GB/s PCIe bridge. The simulator
reports completion in roughly 67 milliseconds. 
\end{example}

\subsection{From Join-Agg Specs to Exchange Programs}
\label{sec:preamble-gen}

As Einsummable's physical optimizer proceeds through the directed,
acyclic graph of abstract operators, for a given abstract operator
\texttt{op}, the optimizer must produce the
preamble for the exchange program used to physically implement
\texttt{op}. To do this, the compiler must determine the atoms that the program will operate over.
These atoms represent the finest pieces of data that the exchange
program can possibly handle so that the exchange program
can do its work by aggregating and routing whole atoms.  Splitting and concatenation happen
only just before or just after the compiled program is executed, in kernel calls generated by
the compiler.

Consider the $k$th input of \texttt{op}, with tensor type
$\mathcal{T}^{(i_k)}$ produced by an already-compiled operator whose
join-agg spec is $\mathcal{J}_k$. Two families of blocks must be considered.
First, the distinct output types of $\mathcal{J}_k$ tile
$\mathcal{T}^{(i_k)}$, and one copy
of each such block per function specification of $\mathcal{J}_k$ that
lists it exists somewhere on a GPU. Second,
\texttt{op}'s join-agg spec $\mathcal{J}$ requires the blocks
$\texttt{in}_k(\mathcal{J})$. These must be delivered, fully
aggregated, to the GPU of every function specification of
$\mathcal{J}$ consuming it. To ensure that no blocks need to be split during the exchange, the compiler forms the \emph{overlay} of
the two. In each dimension, it collects the lower and upper bounds, in
that dimension, of every block in both families. The resulting grid of
tensor types tiles $\mathcal{T}^{(i_k)}$ and its cells are the
atoms for input $k$. Every physically existing
block and every required block is a disjoint union of atoms. 

The compiler brackets the exchange program with local kernel calls.
Before the program runs \texttt{decompose} kernels slice each resident
copy of a producer output block into copies of the atoms it contains. After the program runs, \texttt{recompose} kernels
concatenate the aggregated atoms delivered to each GPU into the
contiguous input blocks expected by the concrete operator $f$.

Given this, the \texttt{IN} declaration lists
one \texttt{ExProg} atom per (atom, copy) pair at the GPU holding
that copy, with size equal to the cardinality of the atom's index set
times the width of the datatype. The \texttt{PARTITIONS} declaration
creates one partition of size one per function specification of
\texttt{op}.  The \texttt{OUT}
declaration requires, for each atom, the $\oplus$-aggregate of its $m$
copies at the partition of every function specification of
\texttt{op} whose input blocks contain that atom. 

\begin{example}[Generating a preamble]
\label{ex:preamble-gen}
Suppose the producer is the matrix multiply of
Example~\ref{ex:matmul-join-agg} compiled with its four-specification
join-agg spec $\{\mathcal{F}_{m,n}\}$, and that the compiler assigned
$\mathcal{F}_{1,1}, \mathcal{F}_{1,2}, \mathcal{F}_{2,1},
\mathcal{F}_{2,2}$ to GPUs $0, 1, 2, 3$. Then GPUs 0 and 1 each hold
an unaggregated copy of the output block $\mathcal{C}_1$ (rows
$[0, 128)$) and
GPUs 2 and 3 each hold a copy of $\mathcal{C}_2$ (rows $[128, 256)$).
The consumer \texttt{op} is a row-wise softmax whose join-agg
spec decomposes \texttt{i} into four intervals of 64 rows, so that
$\texttt{in}_1(\mathcal{J})$ consists of the four blocks with rows
$[0, 64)$, $[64, 128)$, $[128, 192)$, $[192, 256)$, each consumed by
exactly one function specification. The overlay has row split points
$\{0, 64, 128, 192, 256\}$, giving four atoms. Each producer block is the union of two atoms (which we call
$\texttt{X}j$ and $\texttt{Y}j$ for the two copies of atom $j$). The
generated preamble is:
\begin{verbatim}
GPUS: 0, 1, 2, 3
PARTITIONS: S1 {size 1} S2 {size 1}
            S3 {size 1} S4 {size 1}
IN:  X1 @ 0; X2 @ 0; Y1 @ 1; Y2 @ 1;
     X3 @ 2; X4 @ 2; Y3 @ 3; Y4 @ 3
OUT: X1+Y1 @ S1; X2+Y2 @ S2;
     X3+Y3 @ S3; X4+Y4 @ S4
\end{verbatim}
Synthesis chooses
$\texttt{S1}, \ldots, \texttt{S4}$ and the program body.
\end{example}

\subsection{Synthesizing Program Bodies}
\label{sec:transformations}

Given a preamble 
we may synthesize the body in a number of ways. For example,  we may start with a trivial, valid program, and repeatedly
rewrite to a new valid program, using the simulator to cost each candidate.
Regardless of the algorithm that chooses the order of transformations applied, they must be sound (they never produce an invalid
program) and complete (every valid program is reachable, so in
particular a cost-optimal one is). We now define a set of sound and complete transformations.

\begin{definition}[Closure; validity]
\label{def:closure}
A \emph{state} is a set of pairs
$(X, s)$, where $X$ is an aggregate and $s$ is a site. The initial
state is $S_{\texttt{IN}} = \{ (\{\texttt{A}\}, s) :
\texttt{A @ } s \in \texttt{IN} \}$. An action
$X \texttt{ @ } s \texttt{ -> } D$ is \emph{enabled} by a state $S$ if
$S$ contains aggregates $Y_1, \ldots, Y_m$ at site $s$ that partition
$X$ (that is, $X = Y_1 \cup \cdots \cup Y_m$ with the $Y_i$ pairwise
disjoint). Firing it adds $(X, s)$ and $(X, d)$ for each $d \in D$.
The \emph{closure} $S^{*}(P)$ of a program $P$ is the least state that
contains $S_{\texttt{IN}}$ and is closed under firing the enabled
actions of $P$. $P$ is \emph{well-formed} if every action of $P$ is enabled by
$S^{*}(P)$, and \emph{valid} if it is well-formed and
$(X, s) \in S^{*}(P)$ for every requirement $X \texttt{ @ } s$ of
\texttt{OUT}.
\end{definition}

For partitions declared by size alone, a program is valid if some
choice of membership makes it valid; below we treat
the membership as fixed. A
\emph{transformation} is a rule that maps a program to a program, and
it is \emph{sound} if it maps every valid program to a valid program
for the same \texttt{IN} and \texttt{OUT}. Two primitive
transformations suffice for completeness.
$\textsc{Insert}(a)$ maps $P$ to $P \cup \{a\}$ and is permitted
whenever the result is well-formed. $\textsc{Delete}(a)$ maps $P$ to
$P \setminus \{a\}$ and is permitted whenever the result is valid.
Both are sound. Deletion is always sound due to the guard, and insertion because of monotonicity (adding an action can only grow the closure).

\vspace{-3 pt}
\begin{theorem}[Completeness]
\label{thm:completeness}
Let $P$ and $Q$ be valid programs for the same \texttt{IN} and
\texttt{OUT}. Then $P$ can be transformed into $Q$ by
$|Q \setminus P|$ insertions followed by $|P \setminus Q|$ deletions,
with every intermediate program valid.
\end{theorem}

\vspace{-7 pt}
\begin{proof}
Insert the actions of $Q \setminus P$ in a firing order of $Q$;
enabledness is monotone in the program, so each insertion is
permitted, and this reaches $P \cup Q$. Then delete the actions of
$P \setminus Q$ latest-firing first: every intermediate contains
$Q$, so \texttt{OUT} stays satisfied, and every remaining action is
enabled by actions that fire strictly earlier and are still present.
\end{proof}

\vspace{-5 pt}
Since a valid program always exists whenever \texttt{OUT} is
satisfiable at all (gather the atoms of each required aggregate at each
of its target sites and aggregate locally), the theorem implies that all valid
programs are reachable from any valid starting program. 

\textsc{Insert} and \textsc{Delete} are used merely to argue completeness and no reasonable search strategy would use them. We give a more reasonable set of rules below.  Note that these are not unconditional rewrites.
Deleting an action also deletes the materializations it produced, and
those materializations may be needed elsewhere, by an \texttt{OUT}
requirement, or to enable another action.

\begin{itemize}
\item \emph{Multicast fission and fusion:}
\texttt{X @ s ->} $D_1 \cup D_2 \leftrightarrow$ the
pair \texttt{X @ s ->} $D_1$ and
\texttt{X @ s ->} $D_2$, trading egress contention
against action count. Both directions leave the closure unchanged, so
both are always permitted.
\item \emph{Relay introduction and elimination:}
\texttt{X @ s -> d} $\leftrightarrow$
\texttt{X @ s -> r -> d}. This is a ``bounce''
move that routes around slow or contended links. Introduction only
adds a materialization at \texttt{r} and is always permitted;
elimination is permitted only when \texttt{X} at \texttt{r} is not
otherwise required.
\item \emph{Aggregation pushdown and pullup:} the pair
\texttt{X @ r -> s}, \texttt{Y @ r -> s}
$\leftrightarrow$ \texttt{X+Y @ r -> s}. This is the
network analogue of eager versus lazy partial
aggregation~\cite{yan1995eager}. Both directions are guarded:
pushdown deletes the separate materializations of \texttt{X} and
\texttt{Y} at \texttt{s}, while pullup deletes that of \texttt{X+Y}
(a local aggregation \texttt{X+Y @ s} may be inserted to restore
it).
\item \emph{Source substitution:} redirect a transfer of \texttt{X}
to draw from another site at which \texttt{X} is
materialized. Guarded: the redirected transfer no longer materializes
\texttt{X} at the original source.
%\item \emph{Site renaming:} apply a permutation of \texttt{GPUS} to
%the sites of the program body and to the memberships of size-declared
%partitions, guarded by the requirement that the permutation fix the
%site of every \texttt{IN} atom and of every \texttt{OUT} requirement
%at a concrete site. Under this guard the closure renames along with
%the program, so validity is preserved.
\end{itemize}

In our implementation, Definition~\ref{def:closure} ensures the validity of every candidate exchange program that the search
procedure proposes, and Theorem~\ref{thm:completeness}
guarantees the space of valid programs being searched contains a
cost-optimal program.

\section{Experimental Evaluation}
\label{sec:experiments}

Our evaluation asks three questions. First, can our
automatic system meet the performance
of mature systems expertly tuned? Second, how effective is Einsummable's optimization: how important is cost-based selection of
join-agg specs (Section~\ref{sec:logical}), and how much does the synthesis of exchange programs matter
(Section~\ref{sec:physical})? Third, is the communication-cost proxy
of Section~\ref{sec:cost-model} reasonable?

Einsummable is implemented in C++ using CUDA. Concrete operators
are high-performance kernels: cuBLAS matrix multiplies,
FlashAttention-style attention kernels, and Triton-generated
kernels for elementwise operations, aggregation, and the
decomposition/recomposition of atoms
(Section~\ref{sec:preamble-gen}). Every plan is compiled to a
CUDA graph and executed directly. Exchange-program
 synthesis uses an off-the-shelf commercial LLM to propose candidate
programs. Each candidate is checked for validity
(Definition~\ref{def:closure}), costed by the hardware simulator, and
the cheapest is kept. 

Experiments are run on two servers. Most experiments use an NVIDIA
DGX A100 with eight 40 GB A100 GPUs connected by a fully non-blocking
NVSwitch. Experiment 4 additionally uses an NVIDIA DGX V100
with eight 32 GB V100 GPUs, whose hybrid cube-mesh topology
(Section~\ref{sec:exprog}) is challenging. All computations
run in FP16 and are executed with two warm-up rounds followed by ten
timed rounds. 
Figures report mean runtimes in milliseconds with 95\% confidence
intervals (which are vanishingly small).

Two workload families are used throughout. The first is a
LLaMA-scale transformer block \cite{grattafiori2024llama}: model
dimension 4,096, thirty-two query heads and eight key/value heads
(grouped-query attention, as in Example~\ref{ex:flash-sharding}) of
head dimension 128, and FFN dimension 14,336, with attention computed
by fused FlashAttention-style kernels. The block is evaluated on five
token workloads: single sequences of 1K, 8K, and 128K tokens, and
batches of eight and 128 sequences of 1K tokens each. The second
family is five fixed matrix chains, labeled A through E, designed so
that each multiply averages roughly $8192^3$ scalar multiplies while
stressing different plan structures: chains A and B have
\emph{bushy} optimal parenthesizations (three and one independent
sub-products, respectively, verified by the classical matrix-chain
dynamic program); chain C is a three-matrix chain whose optimal order
is right-associated, the opposite of what left-deep data parallelism
produces; and chains D and E have large end matrices multiplied
across thin interior dimensions, so that each plan is forced into a
single large terminal aggregation. The parenthesization is fixed and identical for every system. Writing
$M_i$ for the $i$th matrix of a chain, the parenthesizations are:
Chain~A, $(M_1 M_2)\,\big(((M_3 M_4)((M_5 M_6)(M_7 M_8)))\,M_9\big)$,
nine matrices with dimensions alternating 16,384 and 6,144 between
ends of 10,240; Chain~B, $(M_1 M_2)\,((M_3 M_4)\,M_5)$, five matrices
with dimensions alternating 20,480 and 5,120 between ends of 9,216;
Chain~C, $M_1\,(M_2 M_3)$, with dimensions
$10{,}240 \times 12{,}288 \times 32{,}768 \times 2{,}048$; Chain~D,
$M_1\,\big((M_2 (M_3 (M_4 (M_5 (M_6 (M_7 M_8))))))\,M_9\big)$, nine
matrices with ends of 24,576, necks of 1,536, and a 20,480-wide
interior; and Chain~E,
$M_1\,\big((M_2 (M_3 (\cdots (M_{11} M_{12})\cdots)))\,M_{13}\big)$,
thirteen matrices with ends of 49,152, necks of 1,024, and a
20,480-wide interior.

\begin{table}[t]

\caption{Pearson correlation between predicted communication cost
and measured runtime.}
\vspace{-10 pt}

\label{tab:cost-corr}
\begin{tabular}{lc@{\hspace{2em}}lc}
\toprule
Transformer workload & $r$ & Matrix chain & $r$ \\
\midrule
1 seq $\times$ 1K tokens & 0.72 & Chain A & 0.73 \\
1 seq $\times$ 8K tokens & 0.92 & Chain B & 0.81 \\
1 seq $\times$ 128K tokens & 0.80 & Chain C & 0.87 \\
8 seq $\times$ 1K tokens & 0.85 & Chain D & 0.47 \\
128 seq $\times$ 1K tokens & 0.75 & Chain E & 0.92 \\
\bottomrule
\end{tabular}
\vspace{-15 pt}
\end{table}

\subsection{Experiments Run}

\textbf{Experiment 1:} Comparing Einsummable with mature,
expertly-tuned systems on transformer workloads. We compare
against JAX \cite{bradbury2018jax} (XLA-compiled, with GSPMD sharding
annotations \cite{xu2021gspmd} and cuDNN fused attention), PyTorch
\cite{paszke2019pytorch} (NCCL communication, flash
scaled-dot-product attention kernels, weights replicated and
activations sharded by batch or by sequence, whichever is better for
the workload), and vLLM's FlashAttention-2 kernels
\cite{kwon2023vllm} driven directly. Each baseline was tuned by hand;
Einsummable is fully automatic. All four systems run all five token
workloads on one, four, and eight GPUs, except that JAX was unable to
complete the 128K single-sequence workload. Results for the four
larger workloads are given in Figure~\ref{fig:llama-systems}. On the
1K single-sequence workload, omitted from the figure for space, every
system finishes in 0.8--2.2 ms at eight GPUs (JAX 0.83, Einsummable
1.21, PyTorch 1.72, vLLM 2.19), and it is included in the geometric
means reported below.

\vspace{2 pt}
\noindent
\textbf{Experiment 2:} Comparing Einsummable with mature systems on
matrix-chain arithmetic. On chains A through E we compare against JAX
(JIT-compiled, row-sharded) and PyTorch (cuBLAS multiplies with NCCL
row sharding), again on one, four, and eight GPUs. Results are given
in Figure~\ref{fig:matmul-systems}.

\vspace{2 pt}
\noindent
\textbf{Experiment 3:} Testing the value of coordinated, cost-based
logical optimization. We run the transformer block and the matrix
chains on eight GPUs under three plan selectors, all with physical
optimization on: (a) the full logical optimizer of
Section~\ref{sec:logical} (labeled ``EinDecomp'' in the figures,
after the decomposition algorithm
\cite{DBLP:journals/pvldb/BourgeoisDJLATYYJ25} on which it is based), (b) a \emph{greedy} selector, which is the same dynamic
program with the repartition sub-cost removed (since repartitioning
is the only sub-cost that depends on the input tilings, this
decouples the recurrence and each operator independently chooses its
cheapest join-agg spec), and (c) \emph{random} selection, which chooses
uniformly among valid join-agg specs, with twenty independent draws
per workload. Results are given in
Figures~\ref{fig:llama-modes} and~\ref{fig:matmul-modes}.

\vspace{2 pt}
\noindent
\textbf{Experiment 4:} Testing the value of physical optimization on
a restrictive topology. We compare optimized exchange programs
against naive ones (each aggregated input block is simply gathered
and aggregated at its consuming sites, with no topology-aware
routing, relaying, or placement) for the transformer block and the
matrix chains on the DGX V100, which has a challenging network topology. Results are given in
Figures~\ref{fig:v100-llama} and~\ref{fig:v100-matmul}. For
calibration, we also ran the same comparison on the NVSwitch-based
DGX A100.

\vspace{2 pt}
\noindent
\textbf{Experiment 5:} Validating the communication-cost proxy. The
twenty random plans per workload from Experiment 3 give, for each of
the ten (workload, chain) setups, twenty (predicted cost, measured
runtime) pairs. Because FLOPs are invariant across
decompositions, runtime should be a linear-plus-noise function of
communication, so we report the Pearson correlation
(Table~\ref{tab:cost-corr}).

\subsection{Discussion}

Over all experiments, fully automatic Einsummable met or exceeded
the performance of every system tested. On the transformer block at
eight GPUs, Einsummable's geometric-mean runtime over the five
workloads is 8.97 ms, versus 13.65 ms for hand-tuned PyTorch and
14.87 ms for vLLM.

\begin{figure*}[t]
\centering
\includegraphics[width=\textwidth]{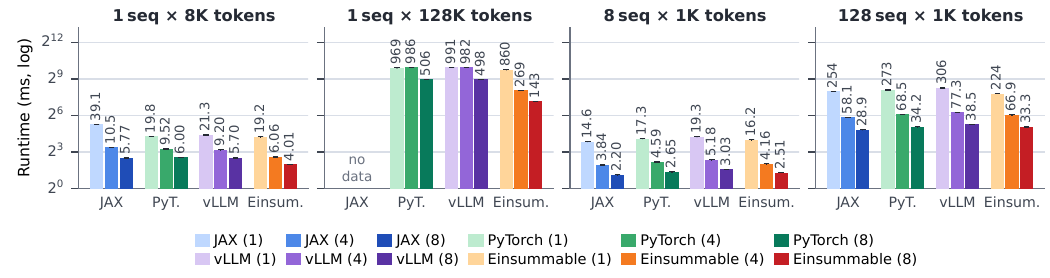}

\vspace{-10 pt}
\caption{Transformer-block mean runtimes (95\% CIs) for Einsummable
vs.\ JAX, PyTorch, and vLLM on one, four, and eight A100 GPUs
(parenthesized numbers in the legend), four LLM workloads.}
\vspace{-10 pt}
\label{fig:llama-systems}
\end{figure*}

\begin{figure*}[t]
\centering
\includegraphics[width=\textwidth]{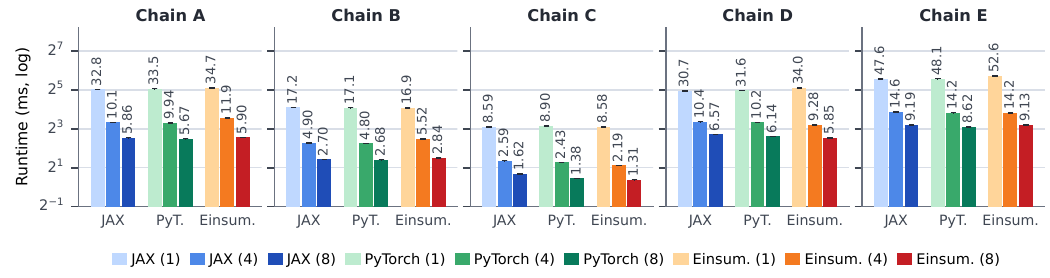}

\vspace{-10 pt}
\caption{Matrix-chain mean runtimes (95\% CIs) for Einsummable vs.\
JAX and PyTorch on one, four, and eight A100 GPUs (parenthesized
numbers in the legend), over chains A to E.}
\vspace{-10 pt}
\label{fig:matmul-systems}
\end{figure*}

\begin{figure*}[t]
\centering
\includegraphics[width=\textwidth]{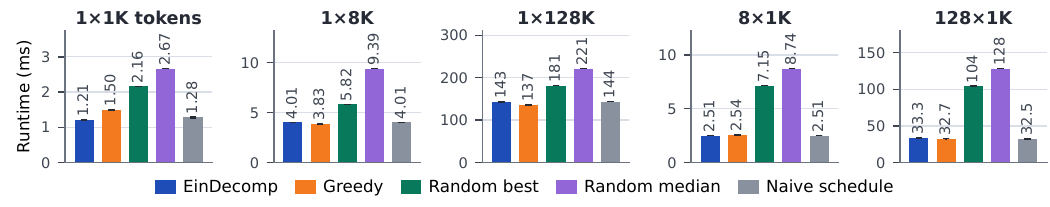}

\vspace{-10 pt}
\caption{Logical-optimization ablation on the transformer block
(eight A100 GPUs): the full optimizer vs.\ greedy, random, and
naive-schedule plan selection.}
\vspace{-10 pt}

\label{fig:llama-modes}
\end{figure*}

\begin{figure*}[t]
\centering
\includegraphics[width=\textwidth]{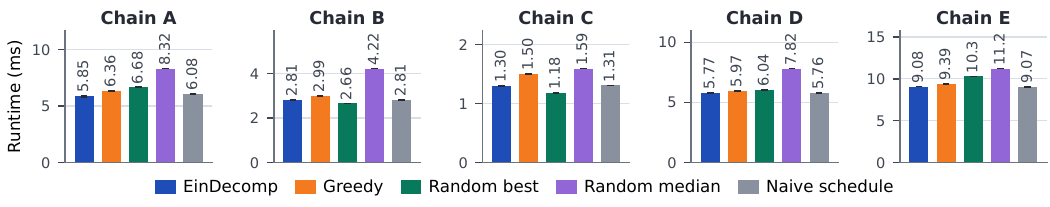}

\vspace{-10 pt}
\caption{Logical-optimization ablation on chains A-E (eight A100
GPUs).}
\vspace{-5 pt}

\label{fig:matmul-modes}
\end{figure*}

\begin{figure}[t]
\centering
\includegraphics[width=\columnwidth]{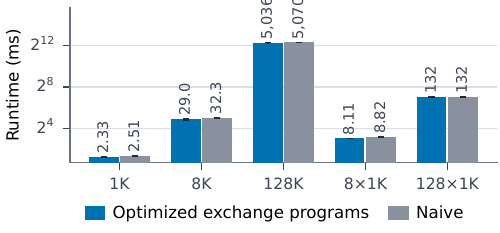}
\vspace{-20 pt}
\caption{Optimized vs.\ naive exchange programs for the transformer
block on the topology-restricted DGX V100.}
\vspace{-10 pt}

\label{fig:v100-llama}
\end{figure}

\begin{figure}[t]
\centering
\includegraphics[width=\columnwidth]{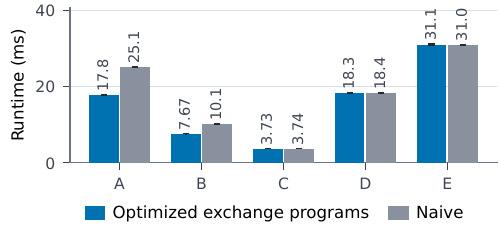}

\vspace{-10 pt}
\caption{Optimized vs.\ naive exchange programs, DGX V100.}
\vspace{-10 pt}

\label{fig:v100-matmul}
\end{figure}

In \textbf{Experiment 1}, the gap is biggest on the long-sequence
workloads. On the single 128K-token sequence, Einsummable requires
143 ms where PyTorch requires 506 ms and vLLM 498 ms (roughly
$3.5\times$ longer). The plan Einsummable selects for this workload
is instructive. A single sequence offers exactly one eight-way
decomposition of the attention operator. Shard \texttt{heads\_kv}, so
that each GPU computes attention for one key/value head and its four
query heads over the entire sequence. The token activations
($131{,}072 \times 4{,}096$, about 1 GB in FP16) dwarf every weight
matrix, so the optimizer shards the normalization, elementwise,
projection, and FFN operators by tokens with replicated weights, and
connects the token-sharded and head-sharded regions by repartitioning
the activations at the boundaries of the attention operator rather
than sharding the contraction of the output projection. This is the
hand-designed DeepSpeed-Ulysses scheme \cite{jacobs2023deepspeed} (sequence parallelism outside attention, head parallelism inside it,
joined by all-to-all exchanges) which Einsummable derives automatically. With a batch of one, the PyTorch implementation can only shard activations by sequence, so every GPU must
gather the full keys and values and absorb the causal-mask load
imbalance across sequence shards.

In \textbf{Experiment 2}, the three systems are effectively tied. At eight
GPUs, every per-chain difference is within about 6\%, and all three
scale by roughly $6\times$ from one to eight GPUs.
We believe that Einsummable's parity is actually a strong result. With no manual decomposition of any kind, it handles classical workloads as well as systems with many person-years of development effort do.

\textbf{Experiment 3} shows that random plan
selection is catastrophic on the transformer block. The median
random plan has a geometric-mean runtime of 22.8 ms against the
optimizer's 8.97 ms ($2.5\times$ worse), and even the best of twenty
draws is 17.6 ms. On
the matrix chains the penalty for a poor logical plan is smaller. The median random plan is
34\% worse in geometric mean, while the best of twenty draws
essentially matches the optimizer (within 3\%, occasionally beating
it on an individual chain). The greedy selector is more interesting. In geometric mean it costs only 2.3\% on the transformer block (and
on a few individual workloads it is slightly faster, within the
noise) but 7\% on the chains.  This may suggest that exchange programs are able to effectively mask the added cost associated with operator-to-operator repartitions.

\textbf{Experiment 4} confirms that physical optimization can be crucial. On the NVSwitch-based A100, where every pair of GPUs enjoys a
dedicated, non-blocking path, optimized and naive exchange programs
are statistically indistinguishable (geometric-mean ratio 1.006 on
the transformer block). With no topology to exploit and the logical
optimizer having already removed unnecessary repartitions, a naive solution does well and there is little improvement possible beyind the naive solution.  However, on the V100's
hybrid cube-mesh, optimized programs are up to
11\% faster on the transformer block (geometric mean 5.6\%), and
$1.41\times$ and $1.32\times$ faster on chains A and B, whose bushy
plans generate cross-socket traffic that the naive programs route through the shared PCIe bridge. On modern hierarchical networks (multi-node NVLink domains,
for example) we expect a similar advantage.

In \textbf{Experiment 5}, across the ten setups, the Pearson correlation
between predicted communication cost and measured runtime ranges
from 0.47 to 0.92, with eight of ten setups above 0.72
(Table~\ref{tab:cost-corr}). One source of inaccuracy seems to be that GPU kernels can have varying performance depending on input shapes even when the floating point operations are identical---modeling this should improve accuracy. Also, the point of the cost model is to select the most inexpensive plan, and we note that the selected
plan was never beaten by any sampled
alternative on the transformer block, indicating that the proxy still ranks plans reasonably.

Finally, a note on the choice of baselines. We compare against
expertly-tuned production systems rather than against prior
auto-parallelizing research systems. Alpa
\cite{zheng2022alpa} can no longer be built against a current JAX stack, and the systems that remain maintained, Unity
\cite{unger2022unity} and nnScaler \cite{lin2024nnscaler}, generate
plans for training rather than for the inference workloads studied here. 

\section{Related Work}
\label{sec:related}

\noindent
\textbf{Automatic parallelization of AI computations.} The named
intra-oper\-ator strategies in wide use, such as tensor (model) parallelism
\cite{shoeybi2019megatron}, sharded data parallelism
\cite{rajbhandari2020zero}, and sequence parallelism
\cite{korthikanti2023sequence}, were each designed by hand.
Auto-parallelizing systems search over combinations of such
strategies: FlexFlow \cite{jia2019flexflow} searches per-operator
parallelization configurations with a randomized strategy, GSPMD
\cite{xu2021gspmd} is a partitioning \emph{mechanism} that propagates
user-supplied sharding annotations through a computation graph and
Alpa \cite{zheng2022alpa} automates both inter-operator (pipeline) and
intra-operator parallelism, the latter via an integer program over
per-operator shardings of tensor dimensions onto a two-dimensional
logical device mesh. Unity \cite{unger2022unity} extends FlexFlow to jointly optimize
algebraic rewrites and parallelization for training.  nnScaler
\cite{lin2024nnscaler} goes furthest in our direction, letting each
operator declare its partitionable dimensions with an einsum-like
annotation; partitioning remains uniform per dimension, aggregation is
fixed to summation, and communication is pattern-matched onto NCCL
collectives. Einsummable's join-agg specs are instead produced by
operator code, admitting non-uniform, data-dependent decompositions
and arbitrary $\oplus$, and its communication is synthesized for
the topology.

\vspace{2 pt}
\noindent
\textbf{Collective communication.} A recent line of work synthesizes topology-aware
\emph{implementations} of the collectives themselves: SCCL
\cite{cai2021sccl} and TACCL \cite{shah2023taccl} search for optimal
algorithms realizing a given collective on a given interconnect, and
Blink \cite{wang2020blink} builds collectives from packing spanning
trees on heterogeneous topologies. Einsummable takes the further step
of having no fixed communication primitives at all. Each exchange
program is a special-purpose communication-and-aggregation pattern.

\vspace{2 pt}
\noindent
\textbf{Relations on tensor runtimes.} Some systems run relational
workloads on tensor infrastructure (TQP \cite{he2022tqp}, Hummingbird
\cite{nakandala2020hummingbird}, Koutsoukos et al.\
\cite{koutsoukos2021tensors}) or bring linear algebra and even
neural-network training into relational engines (SystemML
\cite{ghoting2011systemml}, SystemDS \cite{boehm2020systemds},
Sch{\"u}le et al.\ \cite{schule2024duck}); Einsummable works in the
other direction, using relational optimization for tensor workloads.

\vspace{2 pt}
\noindent
\textbf{Tensor-relational optimization.} Einsummable builds
directly on the tensor-relational model \cite{yuan2021tensor}, which
represents tensors as relations of (key, sub-tensor) pairs so that
kernels do enough work to run efficiently on accelerators, and on the
broader argument that relational systems are a sound substrate for
distributed ML \cite{jankov2019declarative,jankov2021distributed,
tang2023auto}. Our earlier EinDecomp system
\cite{DBLP:journals/pvldb/BourgeoisDJLATYYJ25} takes a first step,
selecting decompositions with a cost-based dynamic program, but only
for computations expressed in an extended Einstein-summation
notation (ruling out FlashAttention, for example). SPORES
\cite{wang2020spores} optimizes linear-algebra expressions via
relational equality saturation, STOREL \cite{schleich2023storel}
optimizes a tensor program given user-supplied storage formats, and
Galley \cite{deeds2025galley} brings cost-based query optimization to
sparse tensor programs. Unlike Einsummable, these target logical rewrites or single-node
kernel structure.

\vspace{2 pt}
\noindent
\textbf{Kernel compilers.} TVM \cite{chen2018tvm}, Triton
\cite{tillet2019triton}, and TACO \cite{kjolstad2017taco} generate
high-performance kernels for tensor operations on a single device.
This work is complementary. Einsummable consumes
kernels (cuBLAS matrix multiplies, FlashAttention
\cite{dao2022flashattention}) as well as kernels
produced by these compilers as its concrete operators.

\vspace{2 pt}
\noindent
\textbf{Parallel database systems.} Intra-operator, partitioned
parallelism with automatic orchestration is the founding idea of
shared-nothing parallel database systems \cite{dewitt1992parallel} and
of Volcano's exchange operator \cite{graefe1990encapsulation}, which
exchange programs generalize (Section~\ref{sec:physical}).

\section{Conclusions}

We have presented Einsummable, which automatically distributes an AI
computation across the GPUs of a server by treating every operation as
a relational join followed by an aggregation over tensor relations.
Operations export their decompositions as join-agg specs, which are selected to minimize
communication, and exchange programs synthesized against the hardware
topology realize the plan without canned collectives.

There is much work to do.  Extending the logical optimization phase so that it is aware of computation cost (and of device properties and how they interact with computational efficiency) would increase accuracy. Extending exchange programs to allow for multi-machine parallelism as well as incorporating learning (to understand how the simulator differs from reality) will be crucial.

\bibliographystyle{ACM-Reference-Format}
\bibliography{references}

\end{document}